\documentclass[10pt]{iopart}

\usepackage{hyperref}
\usepackage{iopams}
\usepackage{nicefrac}
\usepackage{array,multirow,graphicx,pbox}
\usepackage{wrapfig}
\usepackage{pgf,tikz}
\usetikzlibrary{arrows}
\usepackage{fullpage}
\usepackage{cleveref}
\newtheorem{thm}{Theorem}

\newtheorem{defn}{Definition}
\newcommand{\supp}{\mbox{supp}}
\newcommand{\dif}{\mathrm{d}}
\newcommand{\tildeprob}{\stackrel{\textbf{\scriptsize Prob}}{\sim}}
\newcommand{\onprob}{\stackrel{\textbf{\scriptsize Prob}}{\approx}}
\newcommand{\tildeposs}{\stackrel{\textbf{\scriptsize Poss}}{\sim}}
\newcommand{\onposs}{\stackrel{\textbf{\scriptsize Poss}}{\approx}}

\newcommand{\ket}[1]{{\left\vert{#1}\right\rangle}}
\newcommand{\bra}[1]{{\left\langle{#1}\right\vert}}
\newcommand{\ketbra}[1]{\ket{#1}\!\bra{#1}}
\newcommand{\braket}[2]{\left<#1\middle|#2\right>}
\newcommand{\sprod}[2]{\left|\left<#1\middle|#2\right>\right|^2}
\definecolor{xdxdff}{rgb}{0,0,0.}
\definecolor{uuuuuu}{rgb}{0.26666666666666666,0.26666666666666666,0.26666666666666666}

\makeatletter
\DeclareRobustCommand{\qed}{%
  \ifmmode 
  \else \leavevmode\unskip\penalty9999 \hbox{}\nobreak\hfill
  \fi
  \quad\hbox{\qedsymbol}}
\newcommand{\openbox}{\leavevmode
  \hbox to.77778em{%
  \hfil\vrule
  \vbox to.675em{\hrule width.6em\vfil\hrule}%
  \vrule\hfil}}
\newcommand{\qedsymbol}{\openbox}
\newenvironment{proof}[1][\proofname]{\par
  \normalfont
  \topsep6\p@\@plus6\p@ \trivlist
  \item[\hskip\labelsep\itshape
    #1.]\ignorespaces
}{%
  \qed\endtrivlist
}
\newcommand{\proofname}{Proof}
\makeatother

\begin{document}
\title{Contextuality under weak assumptions}
\author{Andrew W. Simmons}
\address{Department of Physics, Imperial College London, London SW7 2AZ, United Kingdom}
\author{Joel J. Wallman}
\address{Institute for Quantum Computing and Department of Applied 
	Mathematics, University of Waterloo, Waterloo, Canada}
\author{Hakop Pashayan}
\address{Centre for Engineered Quantum Systems, School of Physics, The 
University of Sydney, Sydney, NSW 2006, Australia}
\author{Stephen D. Bartlett}
\address{Centre for Engineered Quantum Systems, School of Physics, The 
University of Sydney, Sydney, NSW 2006, Australia}
\author{Terry Rudolph}
\address{Department of Physics, Imperial College London, London SW7 2AZ, United Kingdom}

\begin{abstract}
	The presence of contextuality in quantum theory was first highlighted by Bell, Kochen and Specker, who discovered that for quantum systems of three or more dimensions, measurements could not be viewed as deterministically revealing pre-existing properties of the system. More precisely, no model can assign deterministic outcomes to the projectors of a quantum measurement in a way that depends only on the projector and not the context (the full set of projectors) in which it appeared, despite the fact that the Born rule probabilities associated with projectors are independent of the context. A more general, 
	operational 
	definition of contextuality introduced by Spekkens, which we will term ``probabilistic contextuality'', drops the assumption of determinism and allows for operations other than measurements to be considered contextual. Even two-dimensional quantum mechanics can be shown to be contextual under this generalised notion.  Probabilistic noncontextuality represents the postulate that elements of an operational theory that cannot be distinguished from each other based on the statistics of arbitrarily many repeated experiments (they give rise to the same operational probabilities) are ontologically identical. In this paper, we introduce a framework that enables us to distinguish between different noncontextuality assumptions in terms of the relationships between the ontological representations of objects in the theory given a certain relation between their operational representations.  
This framework can be used to motivate and define a ``possibilistic'' analogue, encapsulating the idea that elements of an operational theory that cannot be unambiguously distinguished operationally can also not be unambiguously distinguished ontologically. We then prove that possibilistic noncontextuality is equivalent to an alternative notion of noncontextuality proposed by Hardy. Finally, we demonstrate that these weaker noncontextuality assumptions are sufficient to prove alternative versions of known ``no-go'' theorems that constrain $\psi$-epistemic models for quantum mechanics.
\end{abstract}

\date{\today}

\maketitle

\section{Introduction}

One way in which quantum mechanics differs strongly from classical mechanics is the existence of incompatible observables;
if we want to think that measurements reveal properties of a system, then we must reconcile this with the fact that there exist pairs of properties
that cannot be simultaneously measured. By considering the same observables appearing in different \emph{contexts}, that is, measured alongside different sets of other observables, Bell\,\cite{Bell2004}, as well as Kochen and Specker\,\cite{Koch1967} showed
that measurements could not be thought of as simply revealing underlying properties of the system in a way that was independent of the context in which the observable was measured. This property of quantum mechanics is now referred to as \emph{contextuality}.

Discussions of contextuality often focus on scenarios in which an element of a operational theory such as quantum mechanics manifests itself in two different contexts, such as two different decompositions of a density matrix; or an observable being measured in two different ways, alongside different sets of co-measurable observables. These manifestations are treated identically by the operational theory, always leading to the same probabilities. In fact, this is why the same notation is used for the objects in the first place, as a context-independent symbol is all that is needed to calculate probabilities.  However there is no formal argument to be made that these elements which are operationally context-independent should also be ontologically context-independent: this must be taken axiomatically.  Motivation for such an axiom cannot be logically deduced or extrapolated from experimental data, but rather must be based on aesthetic principles (e.g., Leibniz's principle\,\cite{Spek2016}). We call such axioms ``noncontextuality assumptions''. This idea, that operationally indistinguishable objects should also be considered ontologically identical is a form of noncontextuality, henceforth referred to as ``probabilistic contextuality''. We see that it somewhat encapsulates the hypothesis, reminiscent of Occam's razor, that these objects give rise to identical properties because they are ontologically identical. We will introduce a framework which allows us to compare and contrast different noncontextuality axioms by viewing them as assumptions about ontological properties motivated by operational data.

This framework can be used to motivate and define a ``possibilistic'' analogue to probabilistic noncontextuality. Possibilistic noncontextuality reflects the assumption that two elements of a physical theory share the same operational possibilities also share the same ontological possibilities. The notion of grouping states by the sets of events that they assign nonzero probability is not entirely novel: this viewpoint emerges naturally in the setting of classical and quantum probability theory from considerations of what it means for different parties to hold different but compatible beliefs about a system such as those of Brun, Finkelstein and Mermin\,\cite{Brun2002} or Caves, Fuchs and Schack\,\cite{Cave2002}.  This notion of possibilistic noncontextuality is strictly weaker than that of probabilistic contextuality. We will demonstrate that analogues of some results originally proved using the stronger notion of probabilistic noncontextuality in fact hold with this weaker assumption: for example, it was shown by Spekkens\,\cite{Spek2005} that the assumption of noncontextuality for preparations is incompatible with the operational predictions of quantum mechanics, and it was shown by Morris\,\cite{Morr2009} and by Chen and Montina\,\cite{Chen2011} that any ontological model obeying noncontextuality for measurements must be $\psi$-ontic. We demonstrate that similar results can be proven using only the weaker, possibilistic, notion of noncontextuality.

\section{Notions of Noncontextuality}

Noncontextuality is a property of an ontological model (also known as a hidden 
variable model).  In order to formulate a generalized 
notion of contextuality, we begin by briefly reviewing the framework of 
ontological models, which allow for realistic descriptions of 
experimental procedures within an operational theory~\,\cite{Spek2005,HR2007}.

\subsection{Ontological models and operational theories}

An operational theory assigns probabilities ${\rm 
Pr}(k|\mathcal{P},\mathcal{T},\mathcal{M})$ to outcomes $k$ occurring when 
procedures for a preparation $\mathcal{P}$, a transformation $\mathcal{T}$, and 
finally a measurement $\mathcal{M}$ are implemented. Any particular $\mathcal{P}$, $\mathcal{M}$ or $\mathcal{T}$ will be denoted an \emph{element} of the operational theory. Quantum mechanics can be 
regarded as an operational theory by identifying preparation procedures with 
density matrices $\rho$, transformations with unitary operators 
$U$ acting on density matrices via conjugation, and $K$-outcome measurements with positive-operator-valued measures (POVM) 
$\mathcal{M} = \{E_1,\ldots E_K\}$. The probability for an outcome $k$ given a preparation described by $\rho$ and a transformation by $U$ is then 
given by the Born rule,
\begin{equation}
{\rm Pr}(k|\rho,U,\mathcal{M}) = \tr [E_k U\rho U^\dagger].
\end{equation}

An ontological model for such an operational theory consists of a set $\Lambda$, equipped with a $\sigma$-algebra $\Sigma$, of 
possible ontic (or ``real'') states alongside ontological representations of preparations, transformations and measurements, dependent on each other only via states $\lambda\in\Lambda$. A 
preparation procedure $\mathcal{P}$ is represented within the ontological model by the preparation of a system in an 
ontic state $\lambda$ sampled according to some measure $\mu_{\mathcal{P}}$ on 
$\Lambda$.  A transformation $\mathcal{T}$ in the operational theory is represented in the ontological model by a map that, for a system in the ontic state
$\lambda$, samples a new $\lambda'\in\Lambda$ according to 
the conditional probability measure $\Gamma_{\mathcal{T}}(\lambda'|\lambda)$ on 
$\Lambda$. Finally, a $K$-outcome measurement $\mathcal{M}$ in the operational theory is represented in the ontological model by a response function $\xi$ such that, for a system in the ontic state $\lambda$, the measurement outcome $k\in\{1,\ldots,K\}=:\mathbb{N}_K$ is sampled from 
the conditional probability distribution $\{\xi_{\mathcal{M}}(k|\lambda)\}$ on 
$\mathbb{N}_K$. 

The ontological model reproduces the predictions of the 
operational theory if
\begin{equation}
{\rm Pr}(k|\mathcal{P},\mathcal{T},\mathcal{M}) = \int \mathrm{d}\lambda\,
\mathrm{d}\lambda'\, \mu_{\mathcal{P}}(\lambda)
\Gamma_{\mathcal{T}}(\lambda'|\lambda) \xi_{\mathcal{M}}(k|\lambda')
\end{equation}
for all $k,\mathcal{P},\mathcal{T},\mathcal{M}$. While all constructions in this 
paper hold for general operational theories, we will only explicitly consider 
ontological models of quantum mechanics. For brevity, we will omit 
transformations and focus solely on preparations and measurements, though all definitions can be directly extended to account for transformations.

The Beltrametti-Bugajski model\,\cite{Belt1995} is perhaps the simplest example 
of an ontological model of pure-state quantum mechanics. In it, pure quantum states are treated as being physical states, so we have $\Lambda=\mathcal{H}$; 
$\mu_\ket{\psi}(\ket{\lambda})=\delta_\ket{\psi}(\ket{\lambda})$; and 
$\xi_{M}(E|\ket{\lambda})=\tr(E\ket{\lambda}\!\bra{\lambda})$. 

We note that for the rest of this paper we will be assuming that our ontological models have a convex structure and therefore preparations (or other operational procedures) can be statistically mixed with each other. Without this assumption many of these noncontextuality relations become trivial. An example of a model without such an assumption is quantum mechanics restricted to pure states, unitary transformations, and projective measurements.

\subsection{Operational and ontological relations}

In this section, we define a generalised \emph{noncontextuality assumption} to be an assumption that 
a particular relation at the level of the operational theory implies a (usually similar) relation at the 
level of the ontological theory. It is a scheme for making conclusions about ontology based on 
a theory's operational predictions. A noncontextuality assumption allows 
reasoning about ontological properties based on observable (operational) 
phenomena.

More formally, an \emph{operational relation} $\sim$ will be a symmetric, 
reflexive relation over elements in the operational theory (i.e., preparations, 
transformations and/or measurements); an \emph{ontological relation} $\approx$ 
is a symmetric, reflexive relation over elements of the ontological model. For 
clarity, $\sim$ and $\approx$ will always be used to denote operational and 
ontological relations respectively.

\begin{defn}[Noncontextuality assumption]
A noncontextuality assumption is a statement that when two objects in the 
operational theory are related by a specified \emph{operational relation} 
$\sim$, their ontological representations must be related by an 
\emph{ontological relation} $\approx$. For example, applied to preparations these are statements
of the form
\begin{equation}
\mathcal{P}_1\sim\mathcal{P}_2 \Rightarrow \mu_{\mathcal{P}_1}\approx \mu_{\mathcal{P}_2}
\end{equation}
\end{defn}

Notably, we do not require that an operational or ontological relation be transitive by definition,
and therefore do not restrict ourselves to only consider equivalence relations for this purpose. In practice, however, 
many interesting noncontextuality assumptions are defined using equivalence relations and elements related by such will be referred to as \emph{equivalent}.

In general, the operational relation is a condition that specifies which 
elements of the operational theory are related to each other. The 
noncontextuality assumption and ontological relation then state how the ontic 
representation of two equivalent elements of the operational theory are related 
to each other. Given a noncontextuality assumption \textbf{X}, we say that an
operational theory exhibits \textbf{X} contextuality if in any ontological model 
that correctly reproduces the operational theory there exists some pair of  operational elements related by $\stackrel{\tiny{\textbf{X}}}{\sim}$ whose ontological representations are not related by $\stackrel{\tiny{\textbf{X}}}{\approx}$.

While noncontextuality assumptions can apply to preparations, transformations and 
measurements, in this section we will use preparations as an example.

\subsection{Probabilistic Noncontextuality}

As discussed, in quantum mechanics there are many elements such as 
density matrices, observables, and POVM elements,
whose operational behaviour is context-independent. The notion of contextuality due to Spekkens\,\cite{Spek2005} captures this tension
by making the assumption that objects that are operationally identical are also ontologically identical. We will present this form of
noncontextuality assumption using our notation here.

\begin{defn}[\textbf{ProbOp}] Two preparation procedures $\mathcal{P}_1$ and 
	$\mathcal{P}_2$ in an operational model are \textit{probabilistically 
	equivalent} (denoted 
	$\mathcal{P}_1\tildeprob \mathcal{P}_2$) if for all outcomes $k$ of all 
	measurement procedures 	$\mathcal{M}$, 
	\begin{equation}
	{\rm Pr}(k|\mathcal{P}_1,\mathcal{M})= 
	{\rm Pr}(k|\mathcal{P}_2,\mathcal{M}).
	\end{equation}
\end{defn}

In the case of quantum mechanics, all operational predictions for a preparation 
procedure $\mathcal{P}$ are completely encoded in the associated density matrix 
$\rho_{\mathcal{P}}$. Therefore we have
\begin{equation}
\mathcal{P}_1 \tildeprob \mathcal{P}_2 \Leftrightarrow \rho_{\mathcal{P}_1} = 
\rho_{\mathcal{P}_1}.
\end{equation}
As a particular example, the maximally mixed state can be prepared in many 
ways, such as by tracing over half of a maximally-entangled pair, randomly 
applying a unitary operation or by preparing a system in a basis and forgetting 
which element was prepared. All of these methods are probabilistically 
equivalent for a single qubit.

Having specified an operational relation, the next step is to specify an 
ontological relation.

\begin{defn}[\textbf{ProbOn}] Two preparation distributions
$\mu_{\mathcal{P}_1}$, $\mu_{\mathcal{P}_2}$ in an ontological model are 
\emph{probabilistically equivalent}, denoted 
$\mu_{\mathcal{P}_1}\onprob\mu_{\mathcal{P}_2}$, if and only if 
$\mu_{\mathcal{P}_1}(\lambda)=\mu_{\mathcal{P}_2}(\lambda)$ $\forall\lambda$.
\end{defn}

\begin{defn}[\textbf{Probabilistic Noncontextuality}] For any two preparations 
$\mathcal{P}_1, \mathcal{P}_2$ related by \textbf{ProbOp}, their ontological 
representations $\mu_{\mathcal{P}_1},\mu_{\mathcal{P}_2}$ are related by 
\textbf{ProbOn}. That is, the probabilistic noncontextuality assumption states that
\begin{equation}
\mathcal{P}_1\tildeprob \mathcal{P}_2 \Rightarrow \mu_{\mathcal{P}_1} \onprob \mu_{\mathcal{P}_2}.
\end{equation}
\end{defn}

In the quantum case, this translates to assuming
\begin{equation}
\mu_{\mathcal{P}} = \mu_{\rho_{\mathcal{P}}}.
\end{equation}
So, in a probabilistically noncontextual model, if there are two preparation procedures that cannot be statistically distinguished by any measurements then those preparations lead to identical physical situations. In a probabilistically contextual model, then there are different physical situations that nonetheless lead to identical operational predictions for any measurement: the inability to distinguish them is caused by lack of sufficiently fine-grained measurement. 

We can look to electromagnetism in order to find a classical example of an application of this principle. Na\"{i}vely, one might expect an operational preparation of an electromagnetic system to correspond to a specific magnetic vector potential as an ontic state. However, different vector potentials correspond to the same experimental predictions if they differ by a curl-free vector field. The application of the probabilistic noncontextuality assumption, then, would be to say that all such preparations are really ontologically identical; they result in the same distribution over ontic states. The vector potential is often said to be unphysical because of the property that different potentials lead to identical experimental predictions; this line of argument would lead one to conclude that not only are such distributions over ontic states identical, but that they are a delta function. That is, we quotient the state space to form a new ontic state $\tilde{\Lambda}=\nicefrac{\Lambda}{\sim}$, where $A\sim A'$ if and only if $\nabla(A-A')=0$. This will ensure that any such pair of ontic states have the same representative in $\tilde{\Lambda}$. The assumptions leading to such an ontological refinement are thus strictly stronger than those leading only to probabilistic noncontextuality.

We see that both the operational relation \textbf{ProbOp} and the ontological 
relation \textbf{ProbOn} are \emph{probabilistic} in that they depend on the 
full set of (exact) probabilities associated with the object. These 
probabilistic relations are very restrictive and as a result, $\tildeprob$ is a very selective relation, leading to a restriction on ontological models that is more easily satisfied than that of a less selective relation.  However, those elements that are related 
have a very strong ontological condition applied by \textbf{ProbOn}. 

Contextuality is often considered to be a generalisation of nonlocality; by Fine's theorem \cite{Fine1982}, the factorisability condition that characterises local distributions in a Bell scenario is interchangeable with the assumption of outcome-determinism for ontological models. With this assumption, nonlocality manifests itself as noncontextuality with respect to contexts chosen jointly by a spacelike separated Alice and Bob.  Some notions of contextuality, such as that of Kochen and Specker\cite{Koch1967}; Klyachko, Can, Cetiner, Bincioglu and Shumovsky \cite{Klya2008}; Abramsky and Brandenburger \cite{Abra2011}; and Ac\'{i}n, Fritz, Leverrier and Bel\'{e}n Sainz \cite{Acin2015} assume these deterministic models for their notions of contextuality. Probabilistic noncontextuality is a step more general, not necessarily assuming determinism, but recovering these notions of noncontextuality when we restrict our gaze to outcome-deterministic ontological models. It can be considered a generalised form of the type of nonlocality identified by Bell; a scenario demonstrates probabilistic contextuality if its precise operational probabilities cannot be fully explained by a noncontextual (rather than nonlocal) model. As such, under certain assumptions, probabilistic noncontextuality can be detected by inequalities that are robust to noise, as an analogue of Bell inequalities \cite{Klya2008, Mazu2016}.

\subsection{Possibilistic Noncontextuality}

Another natural choice for both the operational and ontological relations is a \emph{possibilistic} one, in which we consider only the 
possibilities of operational and ontological events, rather than their 
probabilities. Such possibilistic considerations also appear naturally in the setting of consistency conditions for agents beliefs about the state of a system, which have been previously studied in the literature.  In the work of Brun, Finkelstein and Mermin (BFM)\,\cite{Brun2002}, a consistent set of state assignments is a set of density matrices which have some overlap in their operational possibilities. The different parties can assign different density matrices to a system, but as long as there exist some states in the mutual support of all of them, then there is an event that can occur that all of them agree is possible, meaning that their initial state assignments are consistent. In that of Caves, Fuchs and Schack (CFS)\,\cite{Cave2002}, the strongest compatibility criterion that can be applied is to ask whether or not the density matrices assigned to the state by each of the parties are in the same equivalence class, defined in this possibilistic sense. If they are, then all of them agree on which events are possible and impossible, even if they disagree on the precise probabilities to assign to each event. Both of these conditions are naturally possibilistic rather than probabilistic.

Such considerations motivate us to define possibilistic operational and ontological relations.

\begin{defn}[\textbf{PossOp}] Two preparation procedures $\mathcal{P}_1$ and 
	$\mathcal{P}_2$ in an operational theory are \textit{possibilistically equivalent}, denoted 
	$\mathcal{P}_1 \tildeposs \mathcal{P}_2$, if for all outcomes $k$ of all 
	measurement procedures 	$\mathcal{M}$, 
	\begin{equation}
	{\rm Pr}(k|\mathcal{P}_1,\mathcal{M}) = 0 \Leftrightarrow {\rm 
		Pr}(k|\mathcal{P}_2,\mathcal{M}) = 0.
	\end{equation}
\end{defn}

In the quantum case, two preparation procedures $\mathcal{P}_1$ and $\mathcal{P}_2$ are possibilistically equivalent if and only if the density matrices they give rise to have the same kernel (or equivalently, the same support). The kernel, which is spanned by the states with which $\rho$ has no overlap, completely defines the possibilistic structure of measurements.
\begin{equation}
\mathcal{P}_1 \tildeposs \mathcal{P}_2 \Leftrightarrow \ker\rho_{\mathcal{P}_1} 
= \ker\rho_{\mathcal{P}_2}
\end{equation}
where $\ker M = \{v:Mv = 0\}$ is the kernel of $M$. 

\begin{defn}[Support]
The \textit{support} of a measure $\mu$ is the largest set $S(\mu)$ such that every open set which has non-empty intersection with $S(\mu)$ has nonzero measure.
\end{defn}

More succinctly, but a little less accurately, the support $S(\mu)=\{\lambda:\mu(\lambda)>0\}$.

\begin{defn}[\textbf{PossOn}] Two preparation distributions $\mu_{\mathcal{P}_1}$, $\mu_{\mathcal{P}_2}$ in an ontological model are \emph{possibilistically equivalent}, denoted $\mu_{\mathcal{P}_1}\onposs\mu_{\mathcal{P}_2}$, if $\mathcal{S}(\mu_{\mathcal{P}_1}) = 
\mathcal{S}(\mu_{\mathcal{P}_2})$.
\end{defn}

\begin{defn}[\textbf{Possibilistic Noncontextuality}] For any two preparations 
$\mathcal{P}_1, \mathcal{P}_2$ related by \textbf{PossOp}, their ontological 
representations $\mu_{\mathcal{P}_1},\mu_{\mathcal{P}_2}$ are related by 
\textbf{PossOn}. That is,
\begin{equation}
\mathcal{P}_2\tildeposs \mathcal{P}_2 \Rightarrow \mu_{\mathcal{P}_1} \onposs \mu_{\mathcal{P}_2}.
\end{equation}
\end{defn}

In the quantum case, this translates to assuming
\begin{equation}
\ker \rho_{\mathcal{P}_1} = \ker \rho_{\mathcal{P}_2} \Rightarrow 
\mu_{\mathcal{P}_1} \onposs \mu_{\mathcal{P}_2}.
\end{equation}

A fully possibilistic noncontextuality assumption, then, would be that if two preparation procedures are possibilistically equivalent
operationally, then their ontological representations are also possibilistically equivalent.

A possibilistically noncontextual model for quantum mechanics would provide a natural explanation for 
why there exist different preparation procedures that 
cannot be unambiguously discriminated\,\cite{Ivan1987}; such 
preparations lead to ensembles of ontic states that themselves cannot be 
unambiguously distinguished. Therefore, the lack of the ability to perform unambiguous 
discrimination in a theory could be interpreted as property of the distributions over ontic states themselves, 
rather than emerging from a lack of sufficiently fine-grained measurements.

Additionally, we find that the assumption of possibilistic preparation noncontextuality causes the quantum consistency conditions, mentioned above, for density matrix assignment to coincide with classical notions of consistency of beliefs. If we have a collection of agents, each can describe their subjective beliefs about the state of a quantum system via a density matrix $\rho_i$. According to the BFM consistency criterion, these assignments are consistent if the intersection of the supports of the $\rho_i$ is empty; that is, that there is some subspace of the Hilbert space which each party agrees is in the support of their density matrix. Since the assumption of possibilistic noncontextuality for preparations uniquely defines a set of ontic states associated with this agreed support, this quantum compatibility condition reduces exactly to the classical notion of ``strong consistency'', which is that there exists some physical state that each agent agrees could represent the true state of the system. The strongest compatibility condition considered by Caves, Fuchs, and Schack, which they denote \emph{equal supports}, is that the supports of each of the $\rho_i$ are identical; in the same way as above, the assumption of possibilistic preparation noncontextuality means this condition reduces to the classical concept of \emph{concordance}: that the agents agree on which states are possible and impossible, although their actual probability assignments may differ.

\subsection{Hardy Noncontextuality}

We can now bring other previously-studied notions of noncontextuality within this framework. One such notion is known as ``Hardy'' or ``logical'' noncontextuality \cite{Abra2011,Lalp,Lian2011}, and when applied to preparations corresponds to the assumption that two preparations leading to identical operational predictions must be compatible with the same set of ontic states. As with probabilistic noncontextuality, Hardy noncontextuality has been considered previously under the assumption of outcome determinism, such as in the treatment of Abramsky and Brandenburger. Again, we study a more general notion without this restriction. While probabilistic noncontextuality can be considered a generalisation of Bell nonlocality, Hardy noncontextuality can be considered a generalisation of Hardy's proof of nonlocality. That is, a scenario exhibits Hardy noncontextuality if there are possibilities (rather than detailed probabilities) that cannot be explained within a noncontextual ontological model.

We can formulate Hardy noncontextuality using the operational and ontological relations defined above.

\begin{defn}[\textbf{HardyNC}] An ontological model is Hardy noncontextual iff
$\mathcal{P}_1\tildeprob \mathcal{P}_2\Rightarrow\mu_{\mathcal{P}_1}\onposs \mu_{\mathcal{P}_2}$.
\end{defn}

The assumption of Hardy noncontextuality is strictly weaker than probabilistic noncontextuality 
as it uses the same operational  relation but imposes a weaker ontological relation on the identified elements. Furthermore, 
Hardy noncontextuality is directly implied by possibilistic noncontextuality 
because
\begin{equation}\label{eq:ProbOp_implies_PossOp}
\mathcal{P}_1\tildeprob\mathcal{P}_2
\Rightarrow 
\mathcal{P}_1\tildeposs\mathcal{P}_2.
\end{equation}

Intuitively, the assumption of Hardy noncontextuality appears to also be strictly weaker than 
of possibilistic noncontextuality because the same ontological relation is being 
enforced in both cases, while possibilistic noncontextuality has a operational 
relation that is easier to satisfy and so relates more elements of the theory. 
For example, all density matrices of the form $p\ketbra{0}+(1-p)\ketbra{1}$ for 
$p\in (0,1)$ are \textit{possibilistically} equivalent and yet no two such 
density matrices are \textit{probabilistically} equivalent. 

We now prove, 
however, that Hardy noncontextuality and possibilistic noncontextuality are 
equivalent for ontological models of quantum mechanics. Before doing this, we note that any two preparation procedures that result in the same density matrix being prepared must have the same
ontic supports under an assumption of possibilistic, 
probabilistic, or Hardy noncontextuality. Hence, we will use the notation $\mathcal{S}(\mu_{\rho}) \equiv\mathcal{S}(\mu_{\mathcal{P}})$ for any $\mathcal{P}$ that results in a preparation of $\rho$.

\begin{figure}
\centering\large
\includegraphics[trim={5cm, 16.8cm, 5cm, 2cm},clip,width=0.6\textwidth]{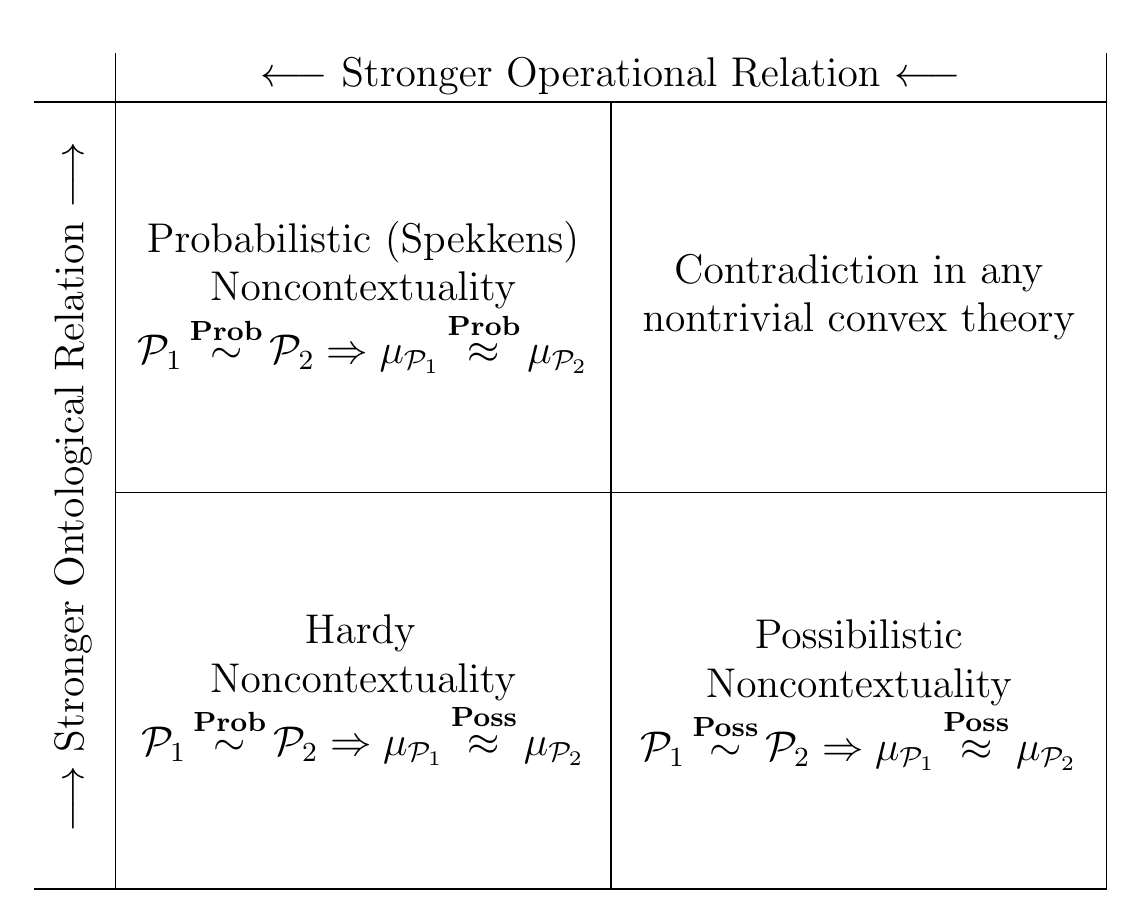}
%

\caption{A square showing the different contextuality assumptions, the strongest of which is contradictory in any convex theory. If convexity is not demanded, there are multiple theories that fulfil the stronger form of noncontextuality indicated as a contradiction above, such as quantum mechanics restricted to pure states and unitary operations.}
\label{fig1}
\end{figure}

\begin{thm}\label{thm:Hardy_poss_eq}
	An ontological model of finite-dimensional quantum mechanics is possibilistically 
	noncontextual if and only if it is Hardy noncontextual.
\end{thm}

\begin{proof}
We will present the proof for preparation procedures here and defer the proofs 
for transformations and measurements to the appendix. From 
\cref{eq:ProbOp_implies_PossOp}, possibilistic preparation 
noncontextuality implies Hardy preparation noncontextuality: any elements related under the probabilistic operational relation are also related under the possibilistic operational relation.  So, we only need to 
prove the converse, that is, that Hardy noncontextuality implies that density 
matrices with the same kernel can be represented by distributions with the same 
support over $\Lambda$.

Let $\rho_0$ and $\rho_1$ be density matrices with the same kernel and with 
smallest nonzero and largest eigenvalues $\alpha_{j,\min}$ and $\alpha_{j,\max}$ for 
$j=0,1$ respectively. We define
\begin{equation}
\sigma_0= \frac{1}{1-\frac{\alpha_{0,\min}}{\alpha_{1,\max}}}\left(\rho_0 - 
\frac{\alpha_{0,\min}}{\alpha_{1,\max}}\rho_1\right).
\end{equation}   
We see that $\sigma_0$ is positive semi-definite, as the largest eigenvalue of 
$\frac{\alpha_{0,\min}}{\alpha_{1,\max}}\rho_{1}$ is 
$\alpha_{0,\min}$ and $\rho_0$ and $\rho_1$ have the same kernel. We note that we must have that $\alpha_{0,\min}\leq\alpha_{1,\max}$, since they have support on subspaces of equal dimensions, say, $d$, and so both $\rho_1$ and $\rho_0$ have $d$ nonzero eigenvalues that sum to 1. Hence, we have $\alpha_{0,\min}=\alpha_{1,\max}$ only in the case that each is the completely mixed state over their support, and $\alpha_{0,\min}<\alpha_{1,\max}$ otherwise. It is easily verified that $\tr(\sigma_0)=1$, so $\sigma_0$ is a density matrix.

Hence $\rho_0$ can be prepared as a convex mixture of $\rho_{1}$ and 
$\sigma_0$ as
\begin{equation}
\rho_0 =\left(1-\frac{\alpha_{0,\min}}{\alpha_{1,\max}}\right) \sigma_0 + 
\frac{\alpha_{0,\min}}{\alpha_{1,\max}}\rho_1.
\end{equation}
Therefore there exists a preparation procedure that prepares 
$\rho_0$ such that 
\begin{equation}
\mathcal{S}(\mu_{\rho_0}) = 
\mathcal{S}(\mu_{\sigma_0)}\cup \mathcal{S}(\mu_{\rho_{1}}) 
\supseteq \mathcal{S}(\mu_{\rho_{1}})
\end{equation}
We can then repeat the argument for $\rho_1$. Therefore $\mathcal{S}(\mu_{\rho_0}) = \mathcal{S}(\mu_{\rho_{1}})$ for 
any $\rho_0$, $\rho_1$ with the same kernel, completing the proof.
\end{proof}
One immediate consequence of this is that it permits the more natural motivation of possibilistic noncontextuality to be applied to the better-known concept of Hardy noncontextuality, providing a motivation for its assumption independent of the fact that it is a weaker noncontextuality assumption than probabilistic noncontextuality. A comparison of these different operational and ontological assumptions can be found in figure \ref{fig1}.

\section{Restrictions on ontological models due to possibilistic noncontextuality}\label{sec:restrictions}
The assumption of probabilistic preparation noncontextuality is incompatible 
with the operational predictions of quantum theory\,\cite{Spek2005}. We now show 
that even the seemingly weaker notion of possibilistic noncontextuality is 
already incompatible with quantum theory. 

\begin{thm}\label{thm:poss_con_4preps}
Any ontological model of quantum mechanics is possibilistically contextual for 
preparations.
\end{thm}

\begin{proof}
Let 
\begin{equation}\label{eq:theta_states}
\ket{\phi} = \cos\left(\frac{\phi}{2}\right) \ket{0} + \sin\left(\frac{\phi}{2}\right) \ket{1}
\end{equation}
for some $0<\phi<\nicefrac{\pi}{2}$. We note that all mixed states are related under the possibilistic operational relation. In particular then, a mixed state made up of $\ketbra{\phi}$ and $\ketbra{-\nicefrac{\pi}{2}}$, and one made up of $\ketbra{-\phi}$ and $\ketbra{\nicefrac{\pi}{2}}$ are related. Therefore, such matrices must have the same support.
\begin{equation}
 \mathcal{S}(\ketbra{-\phi}) \cup\mathcal{S}\left(\ketbra{\frac{\pi}{2}}\right) = 
 \mathcal{S}(\ketbra{\phi})\cup\mathcal{S}\left(\ketbra{-\frac{\pi}{2}}\right).
\end{equation}
Taking the intersection of both sides 
with $\mathcal{S}(\ketbra{\frac{\pi}{2}})$ gives
\begin{eqnarray}\label{eq:steer_sup}
  \mathcal{S}\left(\ketbra{\frac{\pi}{2}}\right) &= &
  \left[\mathcal{S}(\ketbra{\phi}) \cap 
\mathcal{S}\left(\ketbra{\frac{\pi}{2}}\right)\right]\cup
\left[\mathcal{S}\left(\ketbra{-\frac{\pi}{2}}\right)\cap 
\mathcal{S}\left(\ketbra{\frac{\pi}{2}}\right)\right]  \\
&= &\left[\mathcal{S}(\ketbra{\phi}) \cap 
\mathcal{S}\left(\ketbra{\frac{\pi}{2}}\right)\right] \\
&\subseteq& \mathcal{S}(\ketbra{\phi})
\end{eqnarray}
where the second line follows since $\ket{\frac{\pi}{2}}$ and $\ket{-\frac{\pi}{2}}$ are completely operationally distinguishable, and hence must be fully ontologically distinct. This shows that the set of states consistent with $\ketbra{\nicefrac{\pi}{2}}$ is a subset of those consistent with $\ketbra{\phi}$, and so any operational predictions compatible with a preparation of $\ketbra{\nicefrac{\pi}{2}}$must also be consistent with a preparation of $\ketbra{\phi}$. However, 
in order to be consistent with quantum mechanics, preparing $\ket{\frac{\pi}{2}}$ and measuring in the basis 
$\{\ket{\phi},\ket{\pi+\phi}\}$ must give the outcome $\ket{\pi+\phi}$ 
with nonzero probability.
\end{proof}

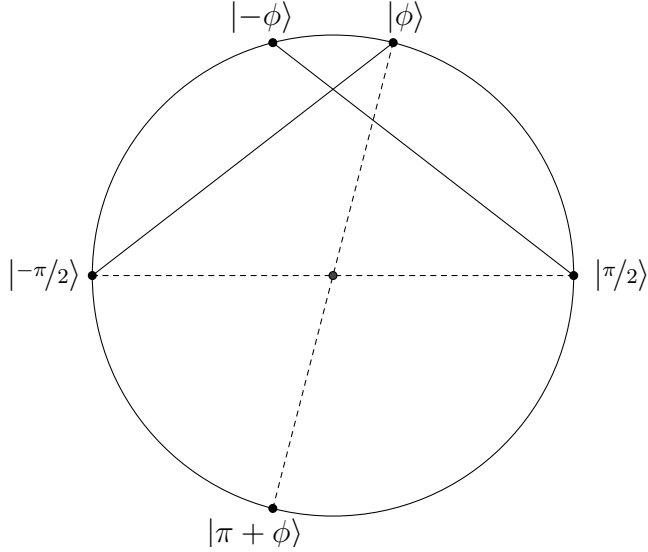
\begin{figure}[h]
\begin{center}
\begin{tikzpicture}[line cap=round,line join=round,>=triangle 45,x=0.8cm,y=0.8cm]
\clip(-5.4,-4.5) rectangle (5.4,5.);

\draw(0.,0.) circle (3.2cm);
\draw (-0.9999999990819624,3.872983346444453)-- (4.,0.);
\draw (0.999999999081962,3.872983346444453)-- (-4.,0.);
\draw [dash pattern=on 2pt off 2pt] (-4.,0.)-- (4.,0.);
\draw [dash pattern=on 2pt off 2pt] (0.9999999990819624,3.872983346444453)-- (-0.9999999990819624,-3.872983346444453);
\begin{scriptsize}
\draw [fill=uuuuuu] (0.,0.) circle (1.5pt);
\draw[color=uuuuuu] (-0.0034097281431338944,-0.25317077583665387) node {};
\draw [fill=xdxdff] (-0.9999999990819624,3.872983346444453) circle (1.5pt);
\draw[color=xdxdff] (-1.1863415516733,4.3) node {\large$\ket{-\phi}$};
\draw [fill=xdxdff] (-0.9999999990819624,-3.872983346444453) circle (1.5pt);
\draw[color=xdxdff] (-1.3,-4.3) node {\large$\ket{\pi+\phi}$};
\draw [fill=xdxdff] (4.,0.) circle (1.5pt);
\draw[color=xdxdff] (4.8,0.) node {\large$\ket{\nicefrac{\pi}{2}}$};
\draw [fill=xdxdff] (0.999999999081962,3.872983346444453) circle (1.5pt);
\draw[color=xdxdff] (1.2114932257527125,4.3) node {\large$\ket{\phi}$};
\draw [fill=xdxdff] (-4.,0.) circle (1.5pt);
\draw[color=xdxdff] (-4.8,0.) node {\large$\ket{\nicefrac{-\pi}{2}}$};
\end{scriptsize}
\end{tikzpicture}
\end{center}
\caption{Diagrams showing states that can be used in the proof of 
\cref{thm:poss_con_4preps}.}
\label{steer}
\end{figure}

The above proof, the states used in which can be seen in figure \ref{steer}, is similar to Hardy's proof of Bell's theorem based only upon 
possibilistic arguments~\,\cite{Hard1993}. In some sense, considering these possibilistic restrictions
on ontological models is forcing them to obey the same logical structure as the operational theory, in this case quantum mechanics.
This is because we can think of working with possibilities as a generalised probability model, in which we work with the Boolean (logical)
semiring in which an event is only prescribed the symbol 0 (impossible) or 1 (possible).


Another known restriction on ontological models under the assumption of probabilistic 
measurement noncontextuality is that such models cannot be $\psi$-epistemic, a 
result first shown by Morris\,\cite{Morr2009} and by Chen and Montina\,\cite{Chen2011}. Here, we will
show an analogous result, in which we use a strictly weaker assumption than probabilistic measurement noncontextuality.

We define two sets associated with a response function for outcome $k$ of a measurement $\mathcal{M}$ in an ontological model:
\begin{eqnarray}
 \mathcal{R}[\xi_M(k|\lambda)] &:=& \{\lambda:\xi_\mathcal{M}(k|\lambda) = 1\} \\
 \mathcal{T}[\xi_M(k|\lambda)] &:=& \{\lambda:\xi_\mathcal{M}(k|\lambda) > 0\}.
\end{eqnarray}
We will show that any ontological model in which these sets depend only upon the POVM element $E_k$ associated with outcome $k$,  and so can be denoted by $\mathcal{R}(E_k)$ and 
$\mathcal{T}(E_k)$ respectively, imposes restrictions on which states can have ontic overlap. We will call such a model \emph{faithful}, as it is a stronger form of the faithfulness condition introduced in\,\cite{Rudo2006}. We note that this condition is implied by probabilistic noncontextuality.

\begin{thm}\label{thm:psi}
Let $\ket{\psi_1}$ and $\ket{\psi_2}$ be two states in a Hilbert space of dimension $d\geq 2^N$, with $\left|\braket{\psi_1}{\psi_2}\right|^2\leq \cos^{2N}(2\tan^{-1}(2^{\nicefrac{1}{2N}}-1))$. Then in any faithful ontological model they have disjoint supports:
\begin{equation}
\mathcal{S}(\mu_{\ket{\psi_1}})\cap\mathcal{S}(\mu_{\ket{\psi_2}})=\emptyset.
\end{equation}
\end{thm}

\begin{proof}
Let $\mathcal{P}_1$ and $\mathcal{P}_2$ be fixed preparation procedures for 
$\ket{\psi_1}$ and $\ket{\psi_2}$ and let $\ket{\psi_1} = \ket{a}^{\otimes 
N}$ and $\ket{\psi_2} = \ket{b}^{\otimes N}$
where $\ket{a}=\ket{\phi}$ and $\ket{b}=\ket{-\phi}$ from 
\cref{eq:theta_states}. 
This form for our states can be made without loss of generality, since for any
pair of states $\ket{\varphi_1}$, $\ket{\varphi_2}$ , there exists a unitary rotation mapping them into a pair of states of our chosen form, and unitary actions must preserve overlaps of support. If $\left|\braket{\varphi_1}{\varphi_2}\right|^2=\left|\braket{\psi_1}{\psi_2}\right|^2$, then there exists a unitary $U$ such that $U\ket{\varphi_1}=\ket{\psi_1}$ and $U\ket{\varphi_2}=\ket{\psi_2}$, and so any state that is in the support of $\ket{\varphi_1}$ and $\ket{\varphi_2}$ is mapped by such a transformation into a state that is in the support of $\ket{\psi_1}$ and $\ket{\psi_2}$ since this is a legitimate preparation procedure for those states.

By considering each of the measurements 
formed by independently measuring $\{\ket{a},\ket{\bar a}\}$ or 
$\{\ket{b},\ket{\bar b}\}$ on each of the $N$ qubits (where $\ket{\bar a}$ denotes the 
orthogonal state to $\ket{a}$, and similarly for $\ket{b}$), we have
\begin{equation}
\mathcal{S}(\mu_{\ket{\psi_1}})\cap\mathcal{S}(\mu_{\ket{\psi_2}}) \subseteq 
\cap_{\vec{\nu}\in\{a,b\}^N} 
\mathcal{R}(\ket{\vec{\nu}}\!\bra{\vec{\nu}}),
\end{equation}
except perhaps on an unphysical set of measure 0, as either $\ket{a}^{\otimes N}$ or $\ket{b}^{\otimes N}$ is orthogonal to all 
but one outcome in each measurement. To see this, consider the case in which we measure
with respect to $\{\ket{a},\ket{\bar a}\}$ on the first qubit, and $\{\ket{b},\ket{\bar b}\}$ on the second.
Here, $\ket{a}\ket{a}$ is incompatible with $\ket{\bar a}\ket{b}$ and $\ket{\bar a}\ket{\bar b}$, with a similar condition holding for $\ket{b}\ket{b}$. The only outcome that is compatible with both $\ket{a}\ket{a}$ and $\ket{b}\ket{b}$ is 
$\ket{a}\ket{b}$.
We note that by assumption, we have $\sprod{a}{b}\leq \cos^{2}(2\tan^{-1}(2^{\nicefrac{1}{2N}}-1))$ and so we can apply the construction in\,\cite{PBR2012}, so there exists a rank-1 projector-valued measure 
(PVM) $M=\{\Pi_{\vec{\nu}}:\vec{\nu}\in\{a,b\}^N\}$ such that $\Pi_{\vec{\nu}} 
\ket{\vec{\nu}}\!\bra{\vec{\nu}} = 0$ for all $\vec{\nu}$~\,\cite{PBR2012}. For 
all $\vec{\nu}$, $\Pi_{\vec{\nu}} \ket{\vec{\nu}}\!\bra{\vec{\nu}} = 0$ and so 
there exist PVMs containing $\Pi_{\vec{\nu}}$ and 
$\ket{\vec{\nu}}\!\bra{\vec{\nu}}$, for example $\{\ketbra{\vec{\nu}},\Pi_{\vec{\nu}}, \mathbb{1} -\ketbra{\vec{\nu}} - \Pi_{\vec{\nu}} \}$. Hence
\begin{equation}
\mathcal{R}(\ket{\vec{\nu}}\!\bra{\vec{\nu}})\cap \mathcal{T}(\Pi_{\vec{\nu}}) 
= \emptyset
\end{equation}
for all $\vec{\nu}$, since a lambda in the intersection of $\mathcal{R}(\ket{\vec{\nu}}\!\bra{\vec{\nu}}$ and $\mathcal{T}(\Pi_{\vec{\nu}})$ would, in a measurement containing both such as the one given above, have to yield the $\ketbra{\vec{\nu}}$ outcome with probability 1, and the $\Pi_{\vec{\nu}}$ outcome with nonzero probability, a clear contradiction.
Therefore 
\begin{eqnarray}
\mathcal{S}(\mu_{\ket{\psi_1}}) \cap\mathcal{S}(\mu_{\ket{\psi_2}})
&\subseteq \left[\cap_{\vec{\nu}\in\{a,b\}^n} 
\mathcal{R}(\ket{\vec{\nu}}\!\bra{\vec{\nu}})\right]\cap \Lambda \\
&\subseteq \left[\cap_{\vec{\nu}\in\{a,b\}^n} 
\mathcal{R}(\ket{\vec{\nu}}\!\bra{\vec{\nu}})\right]\cap 
\left[\cup_{\vec{\nu}\in\{a,b\}^n} 
\mathcal{T}(\Pi_{\vec{\nu}})\right] \\
&\subseteq \emptyset,
\end{eqnarray}
giving the desired result, where the second line follows as some outcome must 
occur when the measurement $M$ is performed.
\end{proof}

Since $\lim_{N\rightarrow\infty}\cos^{2N}(2\tan^{-1}(2^{\nicefrac{1}{2N}}-1))=1$, which can be seen by observing that the expression has a Laurent series expansion at infinity of $1-\nicefrac{1}{x}+O(\nicefrac{1}{x^2})$, any two states can be shown to be ontologically disjoint if we are equipped with a large enough Hilbert space. The result does however differ from those of Morris\,\cite{Morr2009} and of Chen and Montina\,\cite{Chen2011} insofar as this dimensional constraint is present.  As in the PBR theorem\,\cite{PBR2012}, we need a large quantum dimension in order to be able to have a measurement that can distinguish between states with large overlap. We note that in particular, any faithful ontological models for an infinite dimensional space will be $\psi$-ontic, since no two such states can share any ontological support. In general, any quantum system can be thought of as part of a larger quantum system, so results such as this demonstrate an inherent tension between possibilistic measurement noncontextuality and $\psi$-epistemic ontological models. Theorem \ref{thm:psi} can also be thought of as demonstrating that the assumption of preparation independence used by PBR can be replaced by that of faithfulness, as applied here.

We are left to consider which noncontextuality assumptions mandate the property
of faithfulness for ontological models. In particular, is it a consequence of any
of the noncontextuality assumptions we have explored? We shall see that this is
dependent on the requirements one sets on an operational theory. We shall need to introduce the formulation of probabilistic noncontextuality as it is applied to measurements.

\begin{defn}[\textbf{ProbOp-M}] Two measurement outcomes $k^{(1)}_{\mathcal{M}_1}$, belonging to a measurement procedure $\mathcal{M}_1$, and $k^{(2)}_{\mathcal{M}_2}$, belonging to a measurement procedure
$\mathcal{M}_2$, in an operational model are \textit{probabilistically
equivalent} (denoted
$k^{(1)}_{\mathcal{M}_1}\tildeprob k^{(2)}_{\mathcal{P}_2}$) if for all preparation procedures $\mathcal{P}$,
\begin{equation}
{\rm Pr}(k^{(1)}|\mathcal{P},\mathcal{M}_1)=
{\rm Pr}(k^{(2)}|\mathcal{P},\mathcal{M}_2).
\end{equation}
\end{defn}

\begin{defn}[\textbf{ProbOn-M}] Two measurement effects
$\xi_{M_1}(k_1|\lambda)$, $\xi_{M_2}(k_2|\lambda)$ in an ontological model are
\emph{probabilistically equivalent}, denoted
$\xi_{M_1}(k_1|\lambda)\onprob\xi_{M_2}(k_2|\lambda)$, if and only if
$\xi_{M_1}(k_1|\lambda)=\xi_{M_2}(k_2|\lambda)$ $\forall\lambda$.
\end{defn}

\begin{defn}[\textbf{Probabilistic Noncontextuality for Measurements}] For any two preparations
$k^{(1)}_{\mathcal{M}_1}k^{(2)}_{\mathcal{M}_2}$ related by \textbf{ProbOp-M}, their ontological
representations $\xi_{M_1}(k_1|\lambda),\xi_{M_2}(k_2|\lambda)$ are related by
\textbf{ProbOn-M}. That is, the probabilistic noncontextuality assumption states that
\begin{equation}
\mathcal{P}_1\tildeprob \mathcal{P}_2 \Rightarrow \xi_{M_1}(k_1|\lambda) \onprob \xi_{M_2}(k_2|\lambda).
\end{equation}
\end{defn}
We see that the notion of measurement noncontextuality captures two related ideas: that any way of performing a measurement (for example, a specific Naimark dilation for a POVM) yields to ontologically identical predictions; and that the probabilities of specific outcomes are independent of the context in which they appear.

Under this assumption of possibilistic measurement noncontextuality,
first note that both the conditions $\xi_M(k|\lambda)=1$ and $\xi_M(k|\lambda)>0$ can be recognised by considering only the possibilities of alternative outcomes.
By the definition of possibilistic noncontextuality, we then demand that these possibilities, and equally the impossibilities,
are independent of the specific context in which an observable appears. However, for a system in a given ontic state one cannot necessarily conclude that
an outcome certain in one context will be certain in all other contexts in which it appears.
 The noncontextuality assumption taken at face value will only tell us that it cannot be impossible in any other context.
However, we can note that in any context in which this outcome, $E$, appears, we can coarse-grain the other measurement results
by ``forgetting'' which of the alternative outcomes happened. These coarse-grainings yield identical probabilistic predictions,
and in a context in which $E$ is certain, the coarse-graining is impossible. Therefore, these coarse-grainings are impossible in all
contexts, and $E$ is certain in all contexts in which it appears.

The weakness of this argument lies in whether or not we can really consider the coarse-grainings to yield identical probabilistic predictions
with each other. If our scenario is restricted to preparing a quantum state and then measuring it, then this assumption is justified; this is
the position taken for example by Spekkens\,\cite{Spek2005}. However, if we consider that a second measurement procedure could be implemented afterwards, necessitating the existence of a state update rule, then the assumption does not hold and we must use a different notion
of noncontextuality. One such candidate would be a \emph{trichotomistic} notion of noncontextuality, in which our ontological relation, as applied to two related measurements, is that each outcome must be considered impossible, possible-but-not-certain, or certain, independent of the context in which it appears. We can assign these classes the symbols 0, $\nicefrac{1}{2}$ and 1\footnote{This concept of noncontextuality is impossible to replicate within the Abramsky-Brandenburger sheaf-theoretic framework because these values do not form a semiring; coarse-graining two elements that are possible-but-not-certain can be either possible-but-not-certain or certain. Any framework for contextuality based on matching probabilistic predictions in semirings cannot express this sort of notion.}. In any case, this trichotomistic notion of contextuality is also strictly weaker than probabilistic noncontextuality and so a weaker notion is used for this proof regardless of one's position of the equivalence of these coarse-grainings over measurement outcomes.

\section{Tests of possibilistic noncontextuality cannot be robust to experimental error}

Recently, probabilistic noncontextuality inequalities have been demonstrated 
and subsequently shown to be experimentally 
violated\,\cite{Kunj2015,Mazu2016}. In order to avoid finite precision 
loopholes, the probabilistic noncontextuality inequalities are derived assuming 
probabilistic noncontextuality for both preparations and measurements. However, 
as we now prove, no such inequalities can be demonstrated for possibilistic 
noncontextuality.

\begin{thm}
Any operational prediction of quantum mechanics can be approximated arbitrarily well by an ontological model obeying universal possibilistic noncontextuality; that is, the conjunction of possibilistic noncontextuality for preparations and for measurements. 
\end{thm}

\begin{proof}
Let $\rho_\epsilon = (1-\epsilon)\rho + \epsilon I_d/d$ for any fixed 
$\epsilon>0$ and any state $\rho$. For any two density matrices $\rho$ and 
$\sigma$,
\begin{equation}
\rho_\epsilon \tildeposs \sigma_\epsilon.
\end{equation} 
Therefore we can approximate all quantum mechanical predictions by operators 
that are all possibilistically equivalent. The Beltrametti-Bugajski model then 
gives a possibilistic noncontextual model for all states as follows. Let 
$\Lambda = \mathcal{H}$ and $\mu_\ket{\psi_\epsilon}(\ket{\lambda}) = 
(1-\epsilon)\delta(\lambda-\psi) + \epsilon \mu_H$ where $\mu_H$ is the uniform 
Haar measure. We can then extend the model to general mixed states by 
taking convex combinations of $\ketbra{\psi_\epsilon}$, so that 
$\mathcal{S}(\rho_\epsilon) = \mathcal{H}$ for all states $\rho_\epsilon$. The 
same argument also applies to transformations and measurements. 
\end{proof}

Now, any proposed noncontextuality inequality for probabilistic 
noncontextuality will necessarily be a function only of observed test 
statistics compared to some constant value, so we can see that for any 
inequality, there exist universally possibilistically noncontextual models that 
differ from the quantum statistics by an arbitrarily small value. Further, any 
such inequality of some use would have to be saturated by quantum mechanics. 
Thus, finite experimental uncertainty makes impossible the existence of any conclusive proof 
that reality is not universally possibilistically noncontextual. However, were reality probabilistically noncontextual, there would 
be some amount by which the operational predictions would have to be 
altered and this would be experimentally verifiable. Hence, while any 
particular possibilistically noncontextual theory can be experimentally 
refuted, it is not possible to, for example using an analogue of a Bell 
inequality, to refute the entire set of possibilistically noncontextual 
theories. It is worth noting that while there are claims\,\cite{Lund2009} to have an 
experimental verification of Hardy's theorem, these experiments actually verify 
a \emph{probabilistic} version of Hardy's paradox, enabling the existence of 
inequalities that are not $\epsilon$-closely approximated by quantum-mechanical 
predictions.

\section{Further generalisation}\label{FG}
In the previous sections we have developed a framework for expressing arbitrary
examples of noncontextuality relations, and explored the relationships between different forms. 
In this section, we will demonstrate the power of this framework to consider novel forms of contextuality
assumptions, and explore which ontological and operational relations are sensible under certain subjective criteria. 
Up to this point, all of our operational and ontological relations have been equivalence relations, although our definition of such relations does not require transitivity. In fact, there are natural motivations for
notions of noncontextuality that do not have this property.

We could choose an operational relation identifying elements which have an operational distinguishability of no more than 
$\epsilon$: a nontransitive assumption. Likewise, we can have an ontological relation that says that related preparation
distributions should have a
classical distinguishability bounded above by $f(\epsilon)$ for some $f$; 
enforcing a ``similar elements of the theory are represented similarly in the 
ontological model'' idea. An idea of this kind has been explored by 
Winter\,\cite{Wint2014}. 
\begin{defn}[$\stackrel{\epsilon}{\sim}$]
For two preparation procedures $\mathcal{P}_1$ and $\mathcal{P}_2$, we have $\mathcal{P}_1\sim_\epsilon\mathcal{P}_2\Leftrightarrow D(\mathcal{P}_1,\mathcal{P}_2)\leq\epsilon$, where $D(\mathcal{P}_1,\mathcal{P}_2)$ is the operational distinguishability of the two preparations. Assuming quantum theory, this is given by $D(\mathcal{P}_1,\mathcal{P}_2)= \left|\rho_1-\rho_2\right|_1$, where $\rho_1$ and $\rho_2$ are the density operators associated with preparations $\mathcal{P}_1$ and $\mathcal{P}_2$ respectively, and $|\cdot |_1$ is the trace norm.
\end{defn}
\begin{defn}[$\stackrel{\epsilon}{\approx}_f$]
Two preparation distributions $\mu_1, \mu_2$ are related $\mu_1\approx_\epsilon\mu_2$ if and only if $\int_\Lambda\left|\mu_{\mathcal{P}_1}-\mu_{\mathcal{P}_2}\right|\dif\lambda\leq f(\epsilon)$ 
\end{defn}
Up until now, a notion of noncontextuality has consisted of an operational relation, an ontological relation,
and a simple noncontextuality assumption saying that the former implies the latter. Here we see that we have
sets of operational and ontological relations, and our noncontextuality assumption can be thought of as a
kind of axiom schema: for any preparation procedures $\mathcal{P}_1\stackrel{\epsilon}{\sim}\mathcal{P}_2$, we have
$\mu_{\mathcal{P}_1}\stackrel{\epsilon}{\approx}_f\mu_{\mathcal{P}_2}$.

\section{Concluding Remarks}

We have introduced a general framework for the postulation and interpretation of noncontextuality assumptions; we have seen that this view of a noncontextuality assumption as a statement that allows inference of ontological properties from operational ones is both powerful in its descriptive capacity and its ability to highlight novel assumptions.

Using this framework, we have explored weaker notions of contextuality than that of Spekkens\cite{Spek2005}, and we have seen that one such noncontextuality assumption of this kind, possibilistic contextuality, encapsulates the kind of contextuality present in Hardy's proof of Bell's theorem, and that probabilistic contextuality encapsulates the kind present in that of Bell. An interesting open question, then, is to ask what the corresponding contextuality assumptions are that encapsulate Kochen-Specker contextuality, the analogue of the kind present in the GHZ proof of Bell's theorem\,\cite{Gree1990}.

For any scientific realist, the ultimate aim of scientific inquiry is to be able to make statements about the true state of the world, as far as is possible. To be able to make any statements about the ontological nature of the world, we need some sort of noncontextuality assumption in order to allow our operational knowledge to transfer into this domain. Knowing, then, which noncontextuality assumptions are tenable within a given operational scenario, and the relative strengths of these assumptions, is essential in making any claim, however tentative, about the real nature of things.

\section*{Acknowledgements}

AWS would like to thank Angela Xu for her helpful discussions, and acknowledges support from Cambridge Quantum Computing Limited, and from EPSRC, \emph{via} the Centre for Doctoral Training in Controlled Quantum Dynamics. SDB acknowledges support from the ARC \emph{via} the Centre of Excellence in Engineered Quantum Systems (EQuS), project number CE110001013.

\vspace{1cm}

\bibliography{/Users/andrewsimmons/Documents/Latex/Bib/bibliography}{}
\bibliographystyle{iopart-num}

\appendix
\section{}

\label{mainapp}
Just as we defined the support of a preparation procedure above, we will denote the \emph{support of a measurement effect} with response function $\xi(\lambda)$ as
\begin{equation}
\mathcal{R}(\xi) := \{\lambda\in\Lambda | \xi(\lambda)>0\}.
\end{equation}
We can also define a notion of the \emph{support of a transformation procedure} as
\begin{equation}
\mathcal{U}(\Gamma):=\{\{\lambda_1,\lambda_2\}\in\Lambda\otimes\Lambda |\Gamma(\lambda_2|\lambda_1)>0\}.
\end{equation}

\begin{thm}
Possibilistic measurement noncontextuality is equivalent to Hardy measurement noncontextuality.
\end{thm}
\begin{proof}
Again, it is clear that possibilistic measurement noncontextuality implies Hardy measurement noncontextuality, so we only need to prove the converse.

Quantum mechanically, POVM elements are positive semidefinite matrices, and the assumption of possibilistic measurement noncontextuality is the identification of the ontic supports of any two such preparation procedures whose associated density matrices share the same kernel over $\mathbb{R}^2$. The proof here echoes the proof above strongly: if we have two matrices $E_1$ and $E_2$, that share a kernel, and correspond to two measurement effects, we can write $E_1=aE_2+b E_3$, for $E_3$ some matrix with a support strictly contained within that of the other $E_i$. In a similar fashion to the preparations case, we can construct a POVM element with the same matrix as $E_1$ as a convex sum of $E_2$ and $E_3$, demonstrating that $\mathcal{R}(E_1)=\mathcal{R}(E_2)\cup\mathcal{R}(E_3)$. We also have, by symmetry $\mathcal{R}(E_2)=\mathcal{R}(E_1)\cup\mathcal{R}(E_3)$, and therefore $\mathcal{R}(E_1)=\mathcal{R}(E_2)=\mathcal{R}(\mathbb{1}_{\supp(E_i)})$. This completes the proof.
\end{proof}
The above proofs admit a slight generalisation, in that they can show that the possibilistic and probabilistic operational relations lead to identical restrictions for any operational relation based on a function $f(x)$ with $f(x)=f(0)$ iff $x=0$, and $f(x+y)\geq f(x)$, as defined in section \ref{FG}.

We note that in this proof, we are explicitly using the same assumption that we called out in the proof of theorem \ref{thm:psi}, namely that we are considering a coarse-graining of two measurement outcomes via classical post-processing to be possibilistically identical to an ontological coarse graining of the relevant POVM elements. In that section, dropping this implicit assumption weakened our axioms and required we move to an additional noncontextuality assumption. However, for this result, dropping this assumption makes the proof trivial because we are proving a much weaker statement. It can be easily checked that this proof, therefore, holds in both the case in which this assumption is made or not.

\begin{thm}
Possibilistic transformation noncontextuality is equivalent to Hardy transformation noncontextuality.
\end{thm}

\begin{proof}
Once again, it is clear that possibilistic transformation noncontextuality implies Hardy transformation noncontextuality, so we only need to prove the converse.

Consider two transformation procedures $\Gamma_1$ and $\Gamma_2$ with the property that $\forall \xi,\mu$, 
\begin{equation}\int_\Lambda \dif \lambda \dif\lambda' \mu(\lambda)\Gamma_1(\lambda'|\lambda) \xi(\lambda')>0 \Leftrightarrow \int_\Lambda \dif \lambda \dif\lambda' \mu(\lambda)\Gamma_2(\lambda'|\lambda) \xi(\lambda')>0.
\end{equation}
 In general, these are associated with some representations of completely positive, trace preserving maps $T_1$ and $T_2$. Consider their action on the state $\ket{\psi}\bra{\psi}$: they map it to $\sigma_1=T_1\ket{\psi}\bra{\psi}T_1^\dagger$, and $\sigma_2=T_2\ket{\psi}\bra{\psi}T_2^\dagger$ respectively. These two matrices must share a kernel, or else we would be in contradiction with our assumed property.  An example of two maps meeting these criteria might be two dephasing channels with different dephasing strengths.

Note that the most general kind of measurement we can perform to enact tomography on such a transformation procedure is to prepare some entangled state, send part of the entangled state through the transformation procedure, and then follow this up by a joint measurement. By the possibilistic notion of noncontextuality, outcomes of such experiments must yield either a zero probability for both transformations, or a nonzero probability for both transformations. Consider now the Choi-Jamiolkowski isomorphism as applied to our two CPTP maps $T_i$, leading to two channel-states $\tau^{(i)}$. We can see that in general, we require for an entangled initial state $\sigma_{AB}$ and an entangled measurement $E_{AC}$, that
\begin{equation}
\tr\left(\sigma_{AB}\tau^{(1)}_{BC}E_{AC}\right)>0 \quad \Longleftrightarrow \quad  \tr\left(\sigma_{AB}\tau^{(2)}_{BC}E_{AC}\right)>0.
\end{equation}
Taking a trace over $C$ for appropriately chosen $E_{AC}$, \emph{viz} a Bell state $\ketbra{\Phi_+}$ reduces this to the already-proved case of preparation noncontextuality.
\end{proof}

\end{document}